\documentclass[conference]{IEEEtran}
\IEEEoverridecommandlockouts

\usepackage[ruled,linesnumbered]{algorithm2e}
\usepackage{cite}
\usepackage{amsmath,amssymb,amsfonts}
\usepackage{algorithmic}
\usepackage{graphicx}
\usepackage{bm}
\usepackage{float}
\usepackage{subfigure}
\usepackage{textcomp}
\usepackage{xcolor}
\usepackage{amsthm}
\usepackage{breqn}

\newtheorem{definition}{Definition}

\newtheorem{theorem}{Theorem}

\usepackage[backref]{hyperref}
\def\BibTeX{{\rm B\kern-.05em{\sc i\kern-.025em b}\kern-.08em
    T\kern-.1667em\lower.7ex\hbox{E}\kern-.125emX}}
\begin{document}



\title{Learning-based Incentive Mechanism for Task Freshness-aware Vehicular Twin Migration}

\author{
\IEEEauthorblockN{
Junhong Zhang\IEEEauthorrefmark{1}, Jiangtian Nie\IEEEauthorrefmark{2}, Jinbo Wen\IEEEauthorrefmark{1}, Jiawen Kang\IEEEauthorrefmark{1}, Minrui Xu\IEEEauthorrefmark{2}, Xiaofeng Luo\IEEEauthorrefmark{1}, Dusit Niyato\IEEEauthorrefmark{2}, \textit{Fellow, IEEE}
\IEEEcompsocitemizethanks{
The work was supported by NSFC under grant No. 62102099, U22A2054, and the Pearl River Talent Recruitment Program under Grant  2021QN02S643, and also supported in part by National Key R\&D Program of China (No. 2020YFB1807802), and the National Research Foundation (NRF), Singapore and Infocomm Media Development Authority under the Future Communications Research Development Programme (FCP). (Corresponding author: Jiawen Kang (e-mail: kavinkang@gdut.edu.cn)).}
}

\IEEEauthorblockA{
\IEEEauthorrefmark{1}\textit{Guangdong University of Technology, China} \IEEEauthorrefmark{2}\textit{Nanyang Technological University, Singapore}}

}

\maketitle

\begin{abstract}
Vehicular metaverses are an emerging paradigm that integrates extended reality technologies and real-time sensing data to bridge the physical space and digital spaces for intelligent transportation, providing immersive experiences for Vehicular Metaverse Users (VMUs). VMUs access the vehicular metaverse by continuously updating Vehicular Twins (VTs) deployed on nearby RoadSide Units (RSUs). Due to the limited RSU coverage, VTs need to be continuously online migrated between RSUs to ensure seamless immersion and interactions for VMUs with the nature of mobility. However, the VT migration process requires sufficient bandwidth resources from RSUs to enable online and fast migration, leading to a resource trading problem between RSUs and VMUs. To this end, we propose a learning-based incentive mechanism for migration task freshness-aware VT migration in vehicular metaverses. To quantify the freshness of the VT migration task, we first propose a new metric named Age of Twin Migration (AoTM), which measures the time elapsed of completing the VT migration task. Then, we propose an AoTM-based Stackelberg model, where RSUs act as the leader and VMUs act as followers. Due to incomplete information between RSUs and VMUs caused by privacy and security concerns, we utilize deep reinforcement learning to learn the equilibrium of the Stackelberg game. Numerical results demonstrate the effectiveness of our proposed learning-based incentive mechanism for vehicular metaverses.
\end{abstract}


\begin{IEEEkeywords}
Metaverse, vehicular twin, Stackelberg game, Age of Information, deep reinforcement learning.
\end{IEEEkeywords}

\section{Introduction}

The rapid advancement of immersive communication, such as Virtual Reality (VR), Augmented Reality (AR), and ubiquitous Artificial Intelligence (AI) has given rise to the vehicular metaverse. Vehicular metaverses are expected to lead the revolution of intelligent transportation systems by seamlessly blending virtual and physical spaces, allowing for providing immersive services for Vehicular Metaverse Users (VMUs) (i.e., drivers and passengers within vehicles)\cite{zhou2022vetaverse}. Vehicular Twins (VTs) are highly accurate virtual hybrid replicas that cover the entire life cycle of vehicles and VMUs \cite{yu2022bi}. The VTs are updated by sensing data from the surrounding environment to achieve physical-virtual synchronization\cite{xu2022epvisa}. Through VTs, VMUs can access the vehicular metaverse to enjoy a wide range of metaverse applications, such as AR navigation, virtual education, and virtual games\cite{khan2023metaverse, yu2022bi}.

To ensure seamless immersive experiences for VMUs in the vehicular metaverse, resource-limited vehicles offload latency-sensitive and computation-intensive tasks of updating VTs to nearby edge servers in RoadSide Units (RSUs)\cite{yu2022bi}. However, due to the limited coverage of RSUs and the mobility of vehicles, each VT has to be migrated from the current RSU to another to provide uninterrupted immersive services for VMUs. Therefore, the task freshness of the VT migration, i.e., the time it takes to complete the VT migration, is critical to VMUs. To ensure VT migration efficiency, VMUs need to purchase sufficient resources from RSUs for facilitating VT migration, especially bandwidth resources. Without loss of generality, the Metaverse Service Provider (MSP) is set as the manager of RSUs, which is the sole provider of bandwidth resources during VT migration. The MSP aims to optimize its bandwidth selling price and maximize revenue from resource trading with incomplete information. Existing work has been conducted to optimize resource pricing and allocation based on the incentive mechanism in the metaverse\cite{huang2022joint,jiang2022reliable,nguyen2022metachain}. The authors in \cite{huang2022joint} formulated a Stackelberg game joint user association and resource pricing. The authors in \cite{jiang2022reliable} proposed a hierarchical game-theoretic approach to study a reliable coded distributed computing scheme in vehicular metaverses.
However, they ignore the VT migration issue caused by the mobility of vehicles. Therefore, it is still challenging in tackling the resource trading problem in VT migration.

To address the above challenges, in this paper, we propose a new metric named Age of Twin Migration (AoTM) according to the concept of Age of Information (AoI). Considering that VMUs may be reluctant to disclose their private information for privacy security during VT migration, we propose a learning-based incentive mechanism between the MSP and VMUs. The main contributions are summarized as follows:

\begin{itemize}
    \item To quantify the freshness of the VT migration task, we propose a new metric named AoTM according to the concept of AoI for vehicular metaverses and apply it to evaluate the immersion of VMUs.
    \item To improve VT migration efficiency under information incompleteness, we formulate the Stackelberg game between the MSP and VMUs, in which the MSP acts as the leader and VMUs act as followers.
    \item We utilize Deep Reinforcement Learning (DRL) to solve the Stackelberg game under incomplete information. Numerical results demonstrate that the proposed learning-based scheme can converge to the Stackelberg equilibrium and outperform baseline schemes.

\end{itemize}

\section{System Model}

    As shown in Fig. \ref{framework}, edge-assisted remote rendering as a key technology is applied in vehicular metaverses \cite{huang2022joint}. To construct VTs for lower-latency and ultra-reliable metaverse services, such as AR navigation, e-commerce, and virtual games, the large-scale rendering tasks are offloaded to nearby edge servers in RSUs with abundant resources (i.e., storage, bandwidth, and computing) \cite{yu2022bi}. However, due to the dynamic mobility of vehicles and the limited service coverage of RSUs \cite{zhou2022vetaverse}, VTs must be migrated from the source RSUs to the destination RSUs for realizing fully immersive metaverse services. We provide more details of the system model as follows:


\begin{figure}[t]
\vspace{-0.5cm}
\centerline{\includegraphics[width=0.5\textwidth]{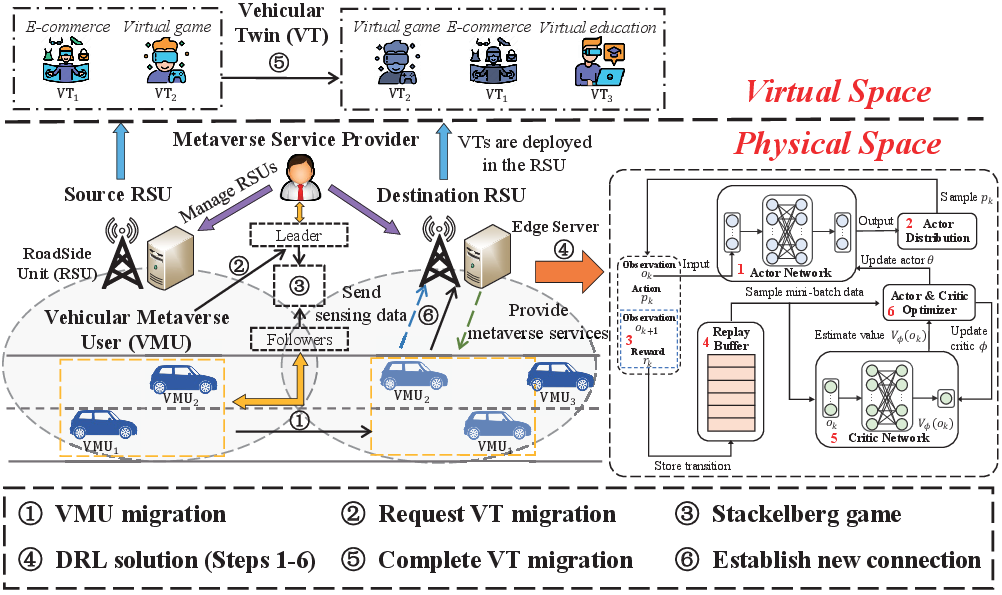}}
\caption{A learning-based incentive mechanism framework for VT migration.}
\label{framework}
\end{figure}

\begin{itemize}
    \item \textbf{MSP:} The MSP as the manager of RSUs can schedule resources of RSUs to provide necessary resources (e.g., computing and bandwidth) for VMUs \cite{huang2022joint}. After being authorized, the MSP can manage a number of communication channels between the source RSUs and the destination RSUs\cite{huang2022joint}. Besides, the MSP leverages sensing data (e.g., traffic conditions and vehicle locations) sent by VMUs to update VTs for providing ultra-reliable and real-time metaverse services for VMUs.
    \item \textbf{VTs:} 
    VTs are the digital replicas deployed in RSUs. They cover the life cycle of vehicles and VMUs and act as intelligent assistants managing metaverse applications\cite{yu2022bi}. In addition, VTs can also analyze and predict their VMUs' behavior through a pre-trained machine learning model. Note that we consider that each VMU has a corresponding VT and the VT can be transmitted in the form of blocks during migration.
    \item \textbf{VMUs:} 
    Without loss of generality, VMUs refer to drivers and passengers within vehicles. The widespread use of VR, AR, and spatial audio devices enables VMUs to enjoy metaverse services through Head-Mounted Displays (HMDs) as well as AR windshields and side windows \cite{zhou2022vetaverse}. Additionally, smart sensors on VMUs (e.g., cameras, Inertial Measurement Units (IMU) suits) collect and send sensing data (e.g., driver fatigue level and vehicle locations) to the MSP for VT synchronization \cite{yu2022bi}.
\end{itemize}

\section{Problem Formulation}
In this section, to quantify the freshness of the VT migration task, we first propose a new metric named AoTM, which can evaluate the immersion of VMUs. Then, we design a Stackelberg game model between the MSP and VMUs for VT migration and analyze the game to prove the existence and the uniqueness of Stackelberg equilibrium among the MSP and VMUs\cite{huang2022joint,zhan2020learning}. In this paper, we consider that one MSP and a set $\mathcal{N}=\{1, \ldots, n, \ldots, N\}$ of $N$ VMUs participate in VT migration and all VTs of VMUs need to be migrated.

\subsection{Age of Twin Migration}
AoI has been widely utilized to quantify data freshness at the destination\cite{yates2021age}. It is defined as the time elapsed since the latest received update was generated at its source, which is a promising metric to improve the performance of time-critical services \cite{kang2022blockchain}. Similarly, in vehicular metaverses, to quantify the freshness of the
VT migration task, we propose a new metric named AoTM according to the concept of the AoI, which is defined as the time elapsed between the last successfully received VT block and the generation of the first VT block in the VT migration. 

We consider that the Orthogonal Frequency Division Multiplexing
Access (OFDMA) technology is applied in the system\cite{huang2022joint}, which ensures that all communication channels occupied by the source RSU and the destination RSU are orthogonal. For VMU $n\in \mathcal{N}$, given the purchased bandwidth $b_n \in (0, +\infty)$ from the MSP, the achievable task transmission rate between the source RSU and the destination RSU is $
\gamma_n = b_n \log_2\Big(1+\frac{\rho h^0 d^{-\varepsilon}}{N_0}\Big),$ where $\rho$, $h^{0}$, $d$, $\varepsilon$, and $N_0$ represent the transmitter power of the source RSU, the unit channel power gain, the distance between the source RSU and the destination RSU, the path-loss coefficient, and the average noise power, respectively\cite{huang2022joint}. Therefore, for VMU $n$, the AoTM of the VT migration task is  
\begin{equation}
    A_n=\frac{D_n}{\gamma_n},
\end{equation}
following the pre-copy live migration strategy in \cite{imran2022live}, the total migrated VT data $D_{n}$ includes the information of system configuration (e.g., CPU and GPU), historical memory data, and real-time states of VMU $n$.

\subsection{Stackelberg Game}
In VT migration, the MSP is the sole bandwidth resource holder and VMUs rely on bandwidth resources provided by the MSP to migrate VTs between RSUs. As a result, a monopoly market is formed, in which the MSP, as the monopolist, has the pricing power of bandwidth and VMUs need to respond to the price by deciding how much bandwidth to purchase. To be specific, when the selling price of bandwidth is low, VMUs may be willing to purchase more bandwidth for enhancing immersive experiences. Conversely, VMUs are reluctant to purchase when the selling price is high, resulting in poor task freshness. Therefore, the selling price of bandwidth has a significant impact on the immersion of VMUs.

To maximize the MSP's profit and maintain its monopoly power, the Stackelberg game can provide a powerful game theoretical model that has been widely used by the monopolist to strategically set the price. The Stackelberg game between the MSP and VMUs consists of two stages. In the first stage, the MSP as the leader decides the selling price of bandwidth for its maximum utility. In the second stage, each VMU as a follower determines the bandwidth demand to maximize its utility. Note that the second stage of the game can be formulated as a competitive game\cite{jiang2022reliable}.

\subsubsection{Utility formation in the VT migration}
 The utility of VMU $n$ is the difference between the profit corresponding to its immersion and its cost of purchasing bandwidth. The higher AoTM impacts the immersive experiences of VMUs negatively, resulting in decreasing the immersion of VMUs\cite{jiang2022reliable}. Following \cite{fan2021cloud}, the immersion function of VMU $n$ obtained from the MSP is defined as
$G_{n} = \alpha_n \ln \left(1+1/A_n\right)$,
where $\alpha_n>0$ is the unit profit for the immersion of VMU $n$.  Therefore, the utility function of VMU $n$ is
\begin{equation}
\label{mu}
U_n(b_n) = G_{n} - p\cdot b_{n},
\end{equation}
where $p>0$ is the unit selling price of bandwidth. In the follower stage, each VMU $n$ maximizes its revenue $U_n(b_n)$ by deciding the best bandwidth demand to purchase. Thus, the problem of maximizing the utility of VMU $n$ is formulated as
\begin{equation}
    \begin{split}
    \textbf{Problem 1:}\: &\max\limits_{b_{n}} \:U_n(b_n)  \\
    &\:\:\text{s.t.}\:\: b_{n} > 0.
    \end{split}
\end{equation}

For the MSP, its utility is the difference between the sum of bandwidth fees paid by all VMUs and the transmission cost for VT migration tasks, which is affected by the unit selling price of bandwidth and bandwidth demands of VMUs. Thus, the utility of the MSP is 
\begin{equation}
U_{s}(p) = \sum_{n=1}^{N}(p\cdot b_n-C\cdot b_n),
\label{msp}
\end{equation}
where $C>0$ is the unit transmission cost of bandwidth for executing the VT migration task, which is proportional to the amount of bandwidth sold to the VMUs. In the first stage, considering that the bandwidth sold by the MSP has a maximum bandwidth $B^{max}$ and the maximum bandwidth pricing $p^{max}$, the MSP maximizes its revenue by deciding a selling price that ensures the total bandwidth sales do not exceed $B^{max}$ and the bandwidth price does not exceed $p^{max}$. Thus, the problem of maximizing the utility of the MSP is formulated as
\begin{equation}
    \begin{split}
    \textbf{Problem 2:}\:&\max\limits_{p}\:U_{s}(p)  \\
    &\:\:\text{s.t.}\:\: {0 < \textstyle \sum_{n=1}^{N}}b_n \leq B^{max},\\
    &\quad\:\:\:\:\:\: b_n > 0,\:\forall n \in \small\{1,\ldots,N\small\},\\
    &\quad\:\:\:\:\:\: 0 < C \leq p \leq p^{max} .
    \end{split}
\end{equation}

\subsubsection{Stackelberg equilibrium analysis}
The Stackelberg game is formulated by combining \textbf{Problem 2} and \textbf{Problem 1}. We seek the Stackelberg equilibrium to obtain the optimal solution to the formulated game. In the Stackelberg equilibrium, the MSP's utility is maximized considering that the VMUs make bandwidth demand strategies based on the best response, and neither the MSP nor any VMU can improve the individual utility by deviating from their strategies \cite{jiang2022reliable,huang2022joint}. The Stackelberg equilibrium is defined as follows:

\begin{definition}
(Stackelberg Equilibrium): We denote $\boldsymbol{b}^{*}=\{b_{n}^{*}\}_{n=1}^N$ and $p^{*}$ as the optimal bandwidth demand strategy vector and the optimal unit bandwidth selling price, respectively. Then, the strategies $(\boldsymbol{b}^{*}=\{b_{n}^{*}\}_{n=1}^N, p^{*})$ can be Stackelberg equilibrium if and only if the following set of inequalities is strictly satisfied:
\begin{equation}
    \left\{\begin{array}{l}U_{s}\left(\boldsymbol{b}^{*},p^{*} \right) \geq U_{s}\left(\boldsymbol{b}^{*}, p  \right),\vspace{0.05in}\\ 
     U_{n}\left(b_{n}^{*} , \boldsymbol{b_{-n}}^{*} , p^{*}\right)\geq U_{n}\left(b_{n} , \boldsymbol{b_{-n}}^{*} , p^{*}\right),\:\forall n \in \mathcal{N}.\end{array}\right.
\end{equation}
\end{definition}

In the following, we adopt the backward induction method to prove the Stackelberg equilibrium\cite{huang2022joint}.

\begin{theorem}
The sub-game perfect equilibrium in the VMUs' subgame is unique.
\end{theorem}

\begin{proof}
We derive the first-order derivative and the second-order derivative of $U_n(b_n)$ with respect to $b_{n}$ as follows:
\begin{equation}
\begin{split}
 \frac{\partial U_{n}(b_n)}{\partial b_n}&=\frac{\alpha_n\log _{2}\left(1+\frac{\rho^{} h^{0} d^{-\varepsilon}}{N_{0}}\right) }{D_n+b_n\log _{2}\left(1+\frac{\rho^{} h^{0} d^{-\varepsilon}}{N_{0}}\right) } -p,\\
\frac{\partial^{2} U_{n}(b_n)}{\partial b_{n}^{2}}&=-\frac{\alpha_n\bigg(\log _{2}\left(1+\frac{\rho^{} h^{0} d^{-\varepsilon}}{N_{0}}\right)\bigg)^2 }{\bigg(D_n+b_n\log _{2}\left(1+\frac{\rho^{} h^{0} d^{-\varepsilon}}{N_{0}}\right)\bigg)^2} <0.
\end{split}
\end{equation}
As the first-order derivative of $U_n(b_n)$ has a unique zero point, and the second-order derivative of $U_n(b_n)$ is negative, the VMU's utility function $U_n(b_n)$ is strictly concave with respect to $b_n$. Then, based on the first-order optimality condition, i.e., $\frac{\partial U_{n}(b_n)}{\partial b_n}=0$, we can obtain the best response function of VMU $n$, given by
\begin{equation}
    b_n^* = \frac{\alpha_n}{p }-\frac{D_n}{\log _{2}\Big(1+\frac{\rho^{} h^{0} d^{-\varepsilon}}{N_{0}}\Big)}.
\label{best respone}
\end{equation}
Therefore, the sub-game perfect equilibrium in the VMUs' subgame is unique.
\end{proof}

\begin{theorem}
There exists a unique Stackelberg equilibrium $(\boldsymbol{b}^*,p^*)$ in the formulated game.
\end{theorem}

\begin{proof}
Based on \textbf{Theorem 1}, the MSP as the leader in the Stackelberg game knows that there exists a unique Nash equilibrium among VMUs under any given value of $p$. Therefore, the MSP can maximize its utility by choosing the optimal $p$. By substituting (\ref{best respone}) into (\ref{msp}), we have
\begin{equation}
    U_{s} = \sum_{n=1}^{N}(p  - C)\Bigg(\frac{\alpha_n}{p}-\frac{D_n}{\log _{2}\big(1+\frac{\rho^{} h^{0} d^{-\varepsilon}}{N_{0}}\big)}\Bigg).
\end{equation}
Then, by taking the first-order derivative and the second-order derivative of $U_{s}(p)$ with respect to $p$, respectively, we have
\begin{equation}
\begin{split}
    \frac{\partial U_{s}(p)}{\partial p} &= \sum_{n=1}^{N}\Bigg(-\frac{D_n}{\log _{2}\big(1+\frac{\rho^{} h^{0} d^{-\varepsilon}}{N_{0}}\big) }+\frac{\alpha _nC}{p^2}\Bigg),\\
    \frac{\partial^{2} U_{s}(p)}{\partial^{2} p}&=\sum_{n=1}^{N}-\frac{2C \cdot \alpha_n}{p^3} <0.
\end{split}
\end{equation}
Since the first-order derivative of $U_{s}(p)$ has a unique zero point, i.e., $p^* = \sqrt{\frac{C\log _{2}\big(1+\frac{\rho^{} h^{0} d^{-\varepsilon}}{N_{0}}\big)\sum_{n=1}^N\alpha_n}{\sum_{n=1}^ND_n}}$, and the second-order derivative of  $U_{s}(p)$ is negative, $U_s(p)$ is strictly concave, indicating that the MSP has a unique optimal solution to the formulated game\cite{zhan2020learning}. Based on the optimal strategy of the MSP, the VMUs' optimal strategies can be obtained\cite{jiang2022reliable}. Therefore, the Stackelberg equilibrium can be obtained uniquely in the formulated game.
\end{proof}

\section{Learning-based Incentive Mechanism with Incomplete Information} \label{Learning-based}
In this section, we first introduce the DRL algorithm. Then, we describe how to transform the Stackelberg game into a learning task. Specifically, we model the Stackelberg game between the MSP and VMUs as a Partially Observable Markov Decision Process (POMDP) and design a DRL-based learning algorithm to explore the optimal solution to the Stackelberg model, where the MSP is the learning agent.

\subsection{Deep Reinforcement Learning for Stackelberg Game} \label{MA-SRL}
Due to the competitive effect, each VMU only has its local information which is incomplete in the game and determines the bandwidth strategies in a fully non-cooperative
manner\cite{huang2022joint}. DRL can be utilized to learn an optimal policy from past experiences based on the current state and the given reward without knowing any prior information. Here are the details of the DRL formulation.


\emph{1) State space:}  At the current game round $k\in \mathcal{K} = \small\{0, \ldots,k,\ldots, K\small\}$, the state space is defined as a union of the current MSP's pricing strategy and VMUs' bandwidth demand strategies, which is denoted as
$ S_k\triangleq \left\{p_k,\boldsymbol{b}_{k}\right\}.$

\emph{2) Partially observable policy:} To tackle the non-stationary problem in the DRL system for facilitating VT migration, we formulate the partially observable space for VT migration. The MSP agent can only make decisions according to its local observation of the environment. We define the observation space $o_k$ of the MSP at the current game round $k$ as a union of its historical pricing strategies and VMUs' bandwidth demand strategies for past $L$ rounds, given by 
\begin{dmath}
    o_k \triangleq \left\{p_{k-L}, \boldsymbol{b}_{k-L}, p_{k-L+1}, \boldsymbol{b}_{k-L+1}, \ldots, p_{k-1}, \boldsymbol{b}_{k-1}\right\}.
\end{dmath}

Note that $p_{k-L}$ and $\boldsymbol{b}_{k-L}$ can be generated randomly during the initial stage when $k<L$. We consider historical information because it enables the MSP agent to learn how its strategy changes impact the game result of the current time slot. When receiving an observation $o_k$, the MSP agent needs to take a pricing action $p_k$ to maximize its utility. Given the lower bound 
cost $C$ and the upper bound price $p^{max}$ for the pricing action, the action space can be represented as $p_k \in[C, p^{max}]$, and the MSP’s policy can be represented as $\pi_{\boldsymbol{\theta}}\left(p_k \mid o_k\right) \rightarrow[C, p^{max}]$. Note that we use a neural network to represent the policy $\pi_{\boldsymbol{\theta}}$ and the value function $V_{\pi_{\boldsymbol{\theta}}}(\cdot)$, where $\boldsymbol{\theta}$ is the neural network parameter.

\emph{3) Reward:} After the state transition, the MSP would gain a reward based on the current state $S_k$ and the corresponding action $p_k$. The reward function of the MSP can be defined as
\begin{equation}
\label{r_MSP}
    R(S_k, p_k)=\begin{cases}1,\:U_s^k \ge U_{best}^k,
\\[2ex]0,\:U_s^k < U_{best}^k,\end{cases}
\end{equation}
where $U_s^k$ is the current utility of the MSP in (\ref{msp}) and $U_{best}^k$ is the highest utility that the MSP has obtained until round $k$.

\emph{4) Value function:} Given a policy $\pi_{\boldsymbol{\theta}}$, the value function $V_{\pi_{\boldsymbol{\theta}}}(S)$ can measure the expected return when starting in $S$ and following $\pi_{\boldsymbol{\theta}}$ thereafter \cite{xu2021multiagent}, which is defined as 
\begin{equation}
    V_{\pi_{\boldsymbol{\theta}}}(S)\triangleq \hat{\mathbb{E}}_{\pi_{\theta}}\left[\sum_{k=0}^{K} \gamma^{k} R\left(S_{k}, p_{k}\right) \mid S_{0}=S\right],
\end{equation}
where $\hat{\mathbb{E}}_{\pi_{\theta}}(\cdot)$ is the expected value of a random variable given that the MSP agent follows the policy $\pi_{\boldsymbol{\theta}}$, and $\gamma \in [0, 1]$ is the reward discounting factor to reduce the weights as the time step increases.

\emph{5) Actor-critic framework design:} We leverage the popular actor-critic framework and the Proximal Policy Optimization 
method for policy iteration \cite{zhan2020learning}. Following \cite{xu2021multiagent}, at each training iteration, we randomly sample experiences from the replay buffer to update the network parameter. Then, \textit{Generalized Advantage Estimation} \cite{schulman2015high} is used to compute variance-reduced advantage function estimator $A(S, p)$ that utilizes a learning state-value function $V_{\pi_{\boldsymbol{\theta}}}(S)$. Since the policy and the value function share the same parameter $\boldsymbol{\theta}$ of the neural network, the loss function consists of the policy surrogate $L^{CLIP}\left(\boldsymbol{\theta}\right)$ and the value function error term $L^{VF}\left(\boldsymbol{\theta}\right)$. Finally, to update the policy and the value function, we utilize stochastic gradient ascent to maximize the objective function as follows:
\begin{dmath}
\label{20}
    \boldsymbol{\theta}_{e+1}=\arg \max _{\boldsymbol{\theta_{e}}} \frac{1}{\left|I\right|} \sum_{\left|I\right|} \hat{\mathbb{E}}_{k}\Big[L_{k}^{CLIP}\left(\boldsymbol{\theta}_{e}\right)-c L_{k}^{V F}\left(\boldsymbol{\theta}_{e}\right)\Big],
\end{dmath}
\begin{equation}
\begin{split}
    L_{k}^{CLIP}(\boldsymbol{\theta}_{e})=\hat{\mathbb{E}}_{k}\bigg[\min &\Big(r_{k}(\boldsymbol{\theta}_{e}) A(S_k,p_k),\\ &f_{clip}\left(r_{k}(\boldsymbol{\theta_{e}})\right)
    A(S_k,p_k)\Big)\bigg],
\end{split}
\end{equation}
\begin{equation}
    L_ {k}^ {VF}( \boldsymbol{\theta}_{e}  )=  \Big(V_{\pi_{\boldsymbol{\theta}_{e}}}(S_ {k})-V_ {k}^ {targ}\Big)^ {2}  ,
\end{equation}
where
\begin{equation}
    r_{k}(\boldsymbol{\theta}_{e})=  \frac {\pi _ {\boldsymbol{\theta}_{e} }(p_ {k}|o_ {k})}{\pi _ {\boldsymbol{\theta}_{e}^{old} }(p_ {k}|o_ {k})}  ,
\end{equation}
\begin{equation}
\begin{split}
    A\left(S_{k}, p_{k}\right)=&-V_{\pi_{\boldsymbol{\theta}_{e}}}\left(S_{k}\right)+\sum_{l=k}^{K-1} \gamma^{l-k} R(S_l,p_l)\\
    &+\gamma^{K-k} V_{\pi_{\boldsymbol{\theta}_{e}}}\left(S_{K}\right),
\end{split}
\end{equation}
and
\begin{equation}
    f_{clip}(r_{k}(\boldsymbol{\theta}_{e}))=  \begin{cases}1-\epsilon,\:r_{k}(\boldsymbol{\theta}_{e})<1-\epsilon,\\
1+\epsilon,\:r_{k}(\boldsymbol{\theta}_{e})>1+\epsilon,\\r_{k}(\boldsymbol{\theta}_{e}),\:1-\epsilon \leq r_{k}(\boldsymbol{\theta}_{e}) \leq 1+\epsilon.\end{cases}  
\end{equation}
Here, $V_k^{targ}$ is the total discount reward from time step $k$ until the end of the episode, $\boldsymbol{\theta}_{e}$ and $\boldsymbol{\theta}_{e+1}$ are the policy parameter in episode $e$ and $e+1$, $\boldsymbol{\theta}_{e}^{old}$ represents the policy parameter for sampling in episode $e$, $c$ is a loss coefficient of the value function, $r_{k}$ is the importance ratio, and $I$ is the batch size of sampled experiences for calculating policy gradients.

\subsection{Algorithm Details}
Motivated by the above analysis, the proposed DRL algorithm details are illustrated in \textbf{Algorithm 1}. 
The time complexity of the proposed DRL algorithm is determined by the multiplication operations in a fully connected deep neural network \cite{zhan2020learning}, which can be expressed as $\mathcal{O}\left(\sum_{f=1}^{F} \epsilon_{f} \epsilon_{f-1}\right)$, where $\epsilon_f$ is the number of neural units in layer $f$ and $F$ is the number of hidden layers.
\begin{algorithm}[t]  
\small
\label{algorithm2}
\caption{Proposed DRL-based Solution for VT Migration}\label{algorithm}
Initialize max round in an episode $K$, number of episodes $E$, batch size $I$ and network parameter $\boldsymbol{\theta}$\;
\For{Episode $e \in 1,\ldots,E $}
{   
    Reset environment state $S_0$ and replay buffer $\mathcal{B F}$\;
    \For{Round $k \in 0,\ldots,K$}
    {   
        MSP observes a state $S_k$ and updates its observation $o_{k-1}$ into $o_k$\;
        Input $o_k$ into MSP's actor policy $\pi_{\boldsymbol{\theta}_e}$ and determine the current price strategy $p_k$\;
        VMUs determine bandwidth demands through (\ref{best respone})\;
        Update $S_k$ into $S_{k+1}$ and calculate reward $R_k$ for the MSP through (\ref{r_MSP}). Then, update $U_{best}^k$ when a higher reward is obtained\;
        Store transition $(o_k,p_k,R_k,o_{k+1})$ into
 $\mathcal{B F}$\;
        \If{ $k\%\left|I\right|==0$}
        {
            \For{$m \in 1,\ldots,M $}
            {
                Sample a random mini-batch of data with a size $\left|I\right|$ from $\mathcal{B F}$ to update the actor and critic through (\ref{20})\;
            }
        }
    }
}
\end{algorithm}

\section{Numerical Results}
In this section, we evaluate the performance of the VT migration system for vehicular metaverses and the proposed DRL-based incentive mechanism through simulation experiments. We first describe the experimental settings, followed by the experimental results and analysis.
\subsection{Experiment Settings}
We consider that there is one MSP and the number of VMUs $N\in \left [ 1,6 \right ] $. Each VT has the data size $D_n\in \left [ 100,300 \right ](\rm{MB}) $ and the immersion coefficient $\alpha_n\in \left [ 5,20 \right ]$. The MSP's maximum bandwidth, transmission cost, and maximum selling price are set to 
$50\rm{MHz}$, $5$, and $50$, respectively. As for the RSU parameters, the transmitter power of the source RSU $\rho $ is $40\rm{dBm}$, the unit channel power gain $h_0$ is $-20\rm{dB}$, the distance between the RSUs $d$ is $500\rm{m}$, the path-loss coefficient $\epsilon$ is $2$, and the average noise power $N_0$ is $-150\rm{dBm}$. The parameters of the DRL are selected through fine-tuning. Specially, we set $L=4$, $D=20$, $E=500$, $K=100$, $M=10$, and $lr=0.00001$ during experiments. Both the two hidden layers of the neural network have $64$ nodes.

\begin{figure}[t]
\vspace{-0.7cm}
\centering 
\subfigure[The return of each episode.]{
\label{convergence.1}
\includegraphics[width=0.48\linewidth]{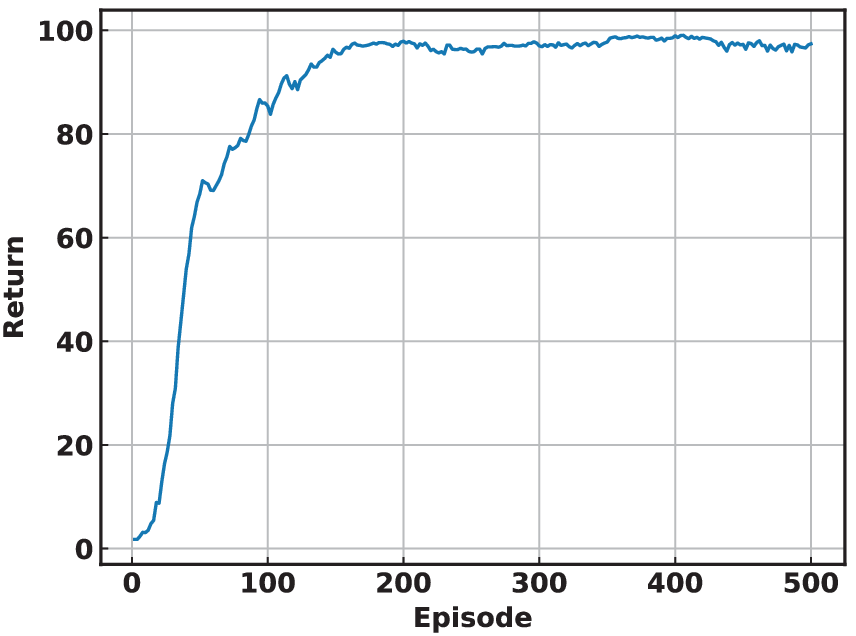}}
\hfill
\subfigure[The utility of the MSP.]{
\label{convergence.2}
\includegraphics[width=0.48\linewidth]{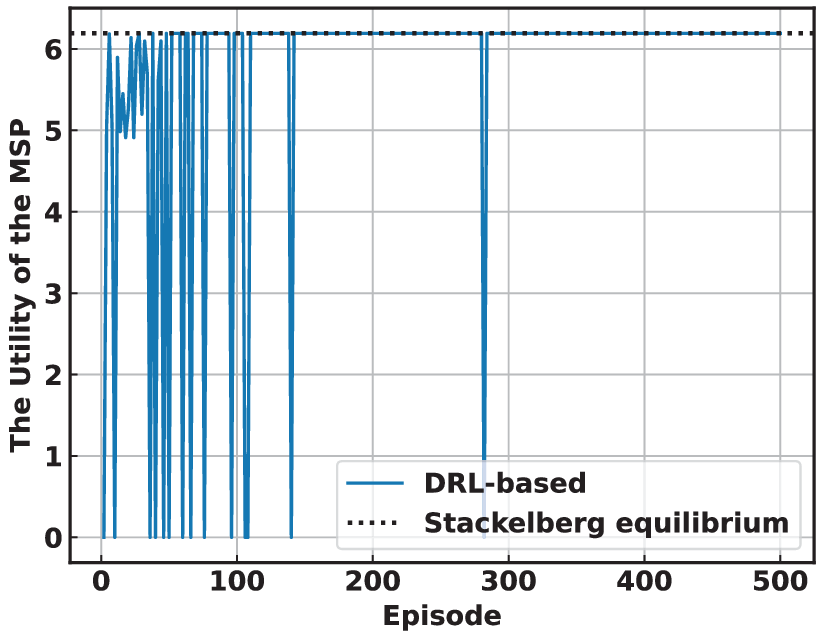}}
\caption{Convergence of DRL-based incentive mechanism.}
\label{convergence}
\end{figure}

\subsection{Experiment Results}
Figure \ref{convergence} shows the convergence of the proposed DRL-based incentive mechanism when there are two VMUs. We set $\alpha_1=\alpha_2=5$, $D_1=200\rm{MB}$, $D_2=100\rm{MB}$, and cost $C=5$. As shown in Fig.~\ref{convergence.1}, the game return of each episode converges to the maximum round, which indicates that the MSP can always choose the optimal strategy in each round. In Fig.~\ref{convergence.2}, the utility of the MSP converges to the Stackelberg equilibrium. Therefore, the DRL-based incentive mechanism under incomplete information is as strong as the Stackelberg game with complete information.

\begin{figure*}[t]
\vspace{-1cm}
	\centering
	\subfigure[ The utility and price strategy of the MSP vs. Transmission cost.]{\label{Cost_MSP}\includegraphics[width=0.24\textwidth]{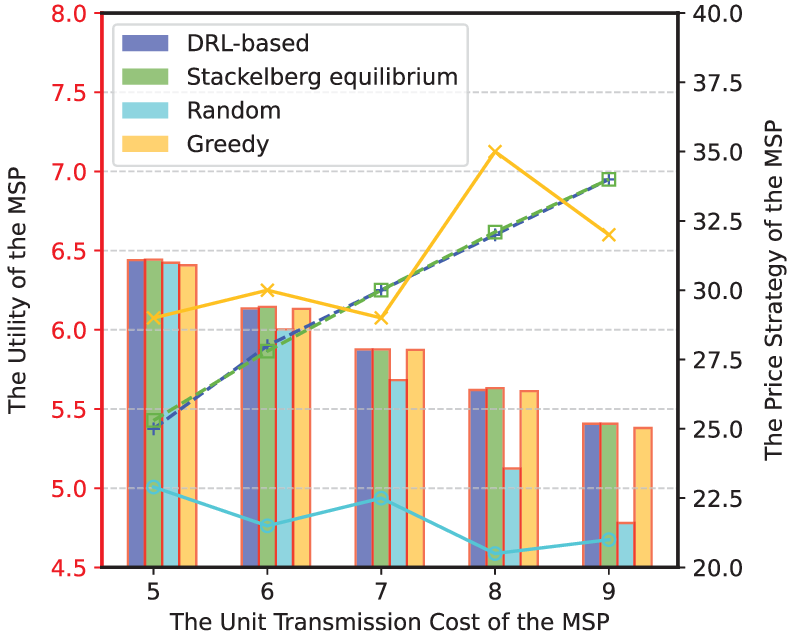}}
	\subfigure[ Total utility and bandwidth strategy of VMUs vs. Transmission cost.]
{\label{Cost_MU}\includegraphics[width=0.24\textwidth]{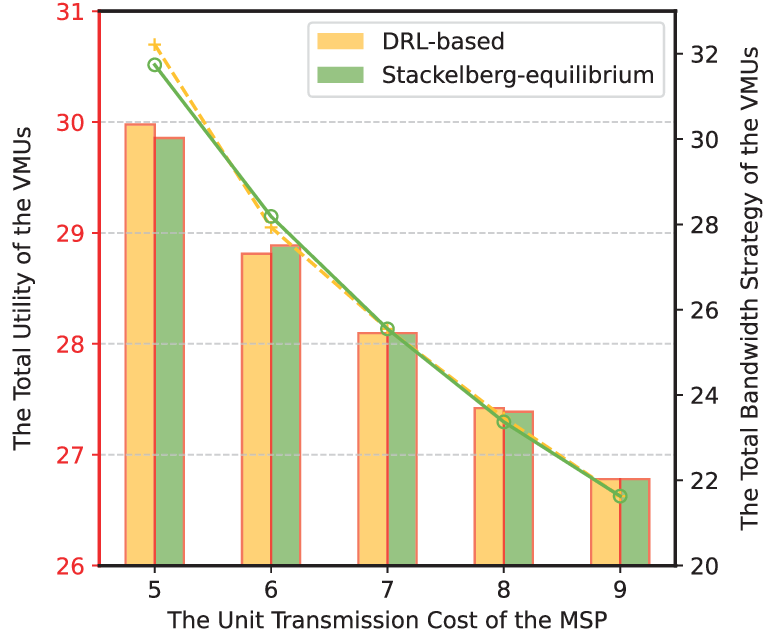}}
	\subfigure[ The utility and price strategy of the MSP vs. Number of VMUs.]
{\label{Num_MSP}\includegraphics[width=0.24\textwidth]{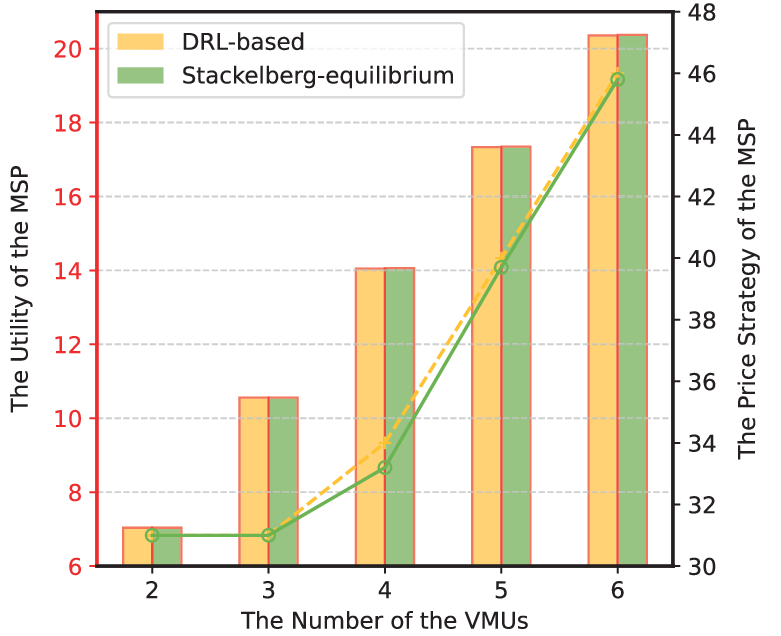}}
	\subfigure[ Average utility and bandwidth strategy of VMUs vs. Number of VMUs.]
{\label{Num_MU}\includegraphics[width=0.24\textwidth]{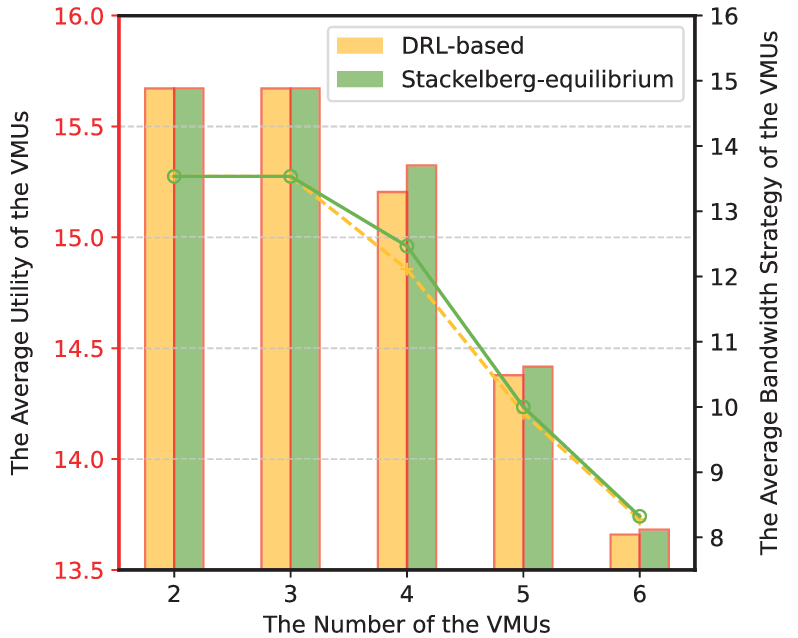}}
	\caption{The performance of the proposed DRL-based incentive mechanism.}\label{performance}
\end{figure*}
Figure \ref{performance} shows the performance of the proposed DRL-based incentive mechanism. In Fig.~\ref{Cost_MSP} and Fig.~\ref{Cost_MU}, we study the influence of the unit transmission cost. Specifically, we study the unit transmission cost by changing it from $5$ to $9$ and consider that there are two VMUs whose VT data sizes are $200\rm{MB}$ and $100\rm{MB}$, and whose immersion coefficients are both $5$. From Fig.~\ref{Cost_MSP} and Fig.~\ref{Cost_MU}, we can see that both the utilities and strategies of the MSP and VMUs in the optimal solutions of the proposed scheme are approaching the Stackelberg equilibrium, which demonstrates that the proposed scheme can find the optimal solution under incomplete information. As the unit transmission cost increases, the pricing of the MSP also increases in Fig.~\ref{Cost_MSP}. For example, when the unit transmission cost is $5$, the MSP sets the price at $25$ to incentive VMUs to perform VT migration. However, when the unit transmission cost is $9$, a higher price of $34$ will be set. In Fig.~\ref{Cost_MU}, we can observe that the total bandwidth strategy of VMUs decreases when the unit transmission cost increases. For example, when the unit transmission cost is $6$, VMUs purchase bandwidth resources of $27.9$. While VMUs only purchase bandwidth resources of $23.4$ when the unit transmission cost is $8$. Both the utilities of the MSP and VMUs significantly decrease due to the high cost of transmission in Fig.~\ref{Cost_MSP} and Fig.~\ref{Cost_MU}. The reason is that when the transmission cost is high, the MSP would increase the bandwidth price due to the cost consideration, leading to a decrease in bandwidth purchased by VMUs because of the high price.
Furthermore, we compare the proposed DRL-based scheme with random and greedy schemes. In the random scheme, the MSP determines the price randomly in each game round, while in the greedy scheme, the MSP determines the best price by selecting from past game rounds. In Fig.~\ref{Cost_MSP}, we can find that our proposed scheme outperforms the baseline schemes.

Next, we study the impacts of the number of VMUs in Fig.~\ref{Num_MSP} and Fig.~\ref{Num_MU}. We set the data size of the VT as $100\rm{MB}$, and the immersion coefficient $\alpha_n$ is $5$. As shown in Fig.~\ref{Num_MSP}, the utility of the MSP increases when the number of VMUs increases. For example, the utility of the MSP is $7.03$ when there are only two VMUs. When the number of VMUs increases to $6$, the MSP can obtain a higher utility of $20.35$. Note that the price of the MSP remains unchanged initially and increases later. The reason is that when there are fewer VMUs, the bandwidth resources of the MSP are sufficient, but when the number of VMUs is too large, the bandwidth of the MSP becomes insufficient. Therefore, the MSP needs to increase the price of bandwidth to limit the purchase of excessive bandwidth by VMUs. As shown in Fig.~\ref{Num_MU}, the average bandwidth purchased by VMUs remains unchanged at first and decreases later. Due to the competition among VMUs, the average utility of VMUs decreased by $12.8\%$ as the number of VMUs increases from $2$ to $6$.

\section{Conclusion}
In this paper, we proposed a learning-based incentive mechanism for task freshness-aware VT migration in vehicular metaverses. To quantify the task freshness of the VT migration, we proposed a new metric called AoTM according to the concept of the AoI. Then, we formulated the resource trading problem between the MSP and VMUs as a Stackelberg game. Furthermore, we utilized DRL to solve the game under incomplete information. Finally, numerical results demonstrate the effectiveness of the proposed mechanism. In the future, we will adopt more effective immersive metrics in conjunction with AoTM to better evaluate the immersion of VMUs and may develop a prototype system to evaluate our framework. Besides, we aim to extend our model to scenarios with multiple MSPs and VMUs.

\bibliographystyle{IEEEtran}
\bibliography{ref}
\end{document}